\documentclass{IEEEtran}
\usepackage{amssymb,amsfonts,amsmath,amsthm}
\usepackage{color}

\usepackage{algorithm}
\usepackage[noend]{algpseudocode}
\usepackage{color}
\usepackage{pgfplots}
\usepackage{filecontents}
\usepackage[export]{adjustbox}
\usepackage{array}

\newtheorem{mydef}{Definition}

\newtheorem{myprop}{Proposition}

\newenvironment{example}
{\let\oldqedsymbol=\qedsymbol
	\renewcommand{\qedsymbol}{$\centerdot$}
	\begin{proof}[\bfseries\upshape Example ]}
	{\end{proof}
	\renewcommand{\qedsymbol}{\oldqedsymbol}}

\begin{document}
\title{\Large \bf Fast Xor-based Erasure Coding based on Polynomial Ring Transforms}
\author{
{\rm Jonathan Detchart,     J\'er\^ome Lacan}\\
ISAE-Supa\'ero, Universit\'e de Toulouse, France
} 

\maketitle

\begin{abstract}
The complexity of software implementations of MDS erasure codes mainly depends on the efficiency of the finite field operations implementation. 
In this paper, we propose a method to reduce the complexity of the finite field multiplication by using simple transforms between a field and a ring to perform the multiplication in a ring. 
We show that moving to a ring reduces the complexity of the operations. Then, we show that this construction allows the use of simple scheduling to reduce the number of operations.
\end{abstract}

\section{Introduction}
\label{sec:intro}
Most of practical Maximum Distance Separable (MDS) packet erasure codes are implemented in software. In the various applications like packet erasure channels \cite{xor-luby} or distributed storage systems \cite{plankScheduling:2011}, the coding/decoding process performs operations over finite fields. The efficiency of the implementation of these finite field operations is thus critical for these applications.  


To speedup this operation, \cite{xor-luby} described an implementation of finite field multiplications which only uses simple \texttt{xor} operations, contrarily to classic software multiplications which are based on lookup tables (LUT). The complexity of multiplying by an element, \emph{i.e.} the number of \texttt{xor} operations, depends on the size of the finite field and also on the element itself. This kind of complexity is studied for Maximum-Distance Separable (MDS) codes in \cite{BlaumR99}. Other work has been done to reduce redundant \texttt{xor} operations by applying scheduling \cite{luoScheduling:2014}.

Independently, in the context of large finite field for cryptographic applications, \cite{ITOH1989} proposed a \texttt{xor}-based method to perform fast hardware implementations of multiplications 
by transforming each element of a field into an element of a larger ring. In this polynomial ring, where the operations on polynomials are done modulo  $x^n+1$,  the multiplication by a monomial is much simpler as the modulo is just a cyclic shift. The authors identified two classes of fields based on irreducible polynomials with binary coefficients 
allowing to transform each field element into a ring element by adding additional "ghost bits".

In this paper, we extend their approach to define fast software implementations of \texttt{xor}-based erasure codes. We propose an original method called PYRIT (PolYnomial RIng Transform) to perform operations between elements of a finite field into a bigger ring by using fast transforms between these two structures. Working in such a ring is much easier than working in a finite field. Firstly, it reduces the coding complexity by design. And secondly, it allows the use of simple scheduling to reduce the number of operations thanks to the properties of the ring structure.

The next section presents the algebraic framework allowing to define the various transforms between the finite field and some subsets of the ring. Then we discuss about the choice of these transforms and their properties. We also detail the complexity analysis before introducing some scheduling results. 

\section{Algebraic context} \label{sec:context}
\label{sec:context}
The algebraic context of this paper is finite fields and ring theory. More detailed presentation of this context including the proofs of the following propositions can be found in \cite{MWSl77} or \cite{poli1992error}.
\begin{mydef}
\label{def:finiteField}
Let  $\mathbb{F}_{q^w}$ be the finite field with $q^w$ elements.  
\end{mydef}
\begin{mydef}
\label{def:ring}
Let $R_{q,n}=\mathbb{F}_{q}[x]/(x^n-1)$ denote the quotient ring of polynomials of the polynomial ring $\mathbb{F}_{2}[x]$ quotiented by the ideal generated by the polynomial $x^n-1$.   
\end{mydef}
\begin{mydef}
\label{def:factorization}
Let $p_1^{u_1}(x)p_2^{u_2}(x)\ldots p_r^{u_r}(x)=x^n-1$ be the decomposition of $x^n-1$ into irreducible polynomials over $\mathbb{F}_{q}$.  
\end{mydef}
When $n$ and $q$ are relatively prime, it can be shown that $u_1=u_2=\ldots=u_n=1$ (see \cite{poli1992error}). In other words, if $q=2$, and $n$ is odd, we simply have $p_1(x)p_2(x)\ldots p_r(x)=x^n-1$.

In the rest of this document, we assume that $n$ and $q$ are relatively prime.
\begin{myprop}
\label{prop:directSumDecomposition}
The ring $R_{q,n}$ is equal to the direct sum of its $r$ minimal ideals of  $A_i=((x^n-1)/p_i(x))$ for $i=1,\ldots,r$. 
\end{myprop}
Moreover, each minimal ideal contains a unique primitive idempotent $\theta_i(x)$. A construction of this idempotent is given in \cite{MWSl77}, Chap. 8, Theorem 6.\\
Since $\mathbb{F}_{q}[x]/(p_i(x))$ is isomorphic to the finite field $ B_i=\mathbb{F}_{q^w_i}$, where $p_i(x)$ is of degree $w_i$, we have:
\begin{myprop}
\label{prop:directProductDecomposition}
$R_{q,n}$ is isomorphic to the following Cartesian product:
$$ R_{q,n} \simeq  B_1 \otimes B_2 \otimes \ldots B_r $$
%
For each $i=1,\ldots,r$, $A_i$ is isomorphic to $B_i$. The isomorphism is:
\begin{equation}
\phi_i  :  \begin{array}{ccc}
B_i & \rightarrow & A_i \\
b(x) & \rightarrow & b(x)\theta_i(x) \
\end{array}
\end{equation}  	

and the inverse isomorphism is:
\begin{equation}
\phi_i^{-1}  :  \begin{array}{ccc}
A_i & \rightarrow & B_i \\
a(x) & \rightarrow & a(\alpha_i) \\
\end{array}
\end{equation}  	
where $\alpha_i$ is a root of $p_i(x)$.
\end{myprop}


Let us  now assume that $q=2$. Let us introduce a special class of polynomials:
\begin{mydef}
	\label{def:AOP}
	The All One Polynomial (AOP) of degree $w$ is defined as 
	$$p(x) = x^w+x^{w-1}+x^{w-2}+\ldots+x+1$$
\end{mydef}

The AOP of degree $w$ is irreducible over $\mathbb{F}_{2}$  if and only if $w+1$ is a prime and $w$ generates $\mathbb{F}^*_{w+1}$, where $\mathbb{F}^*_{w+1}$ is the multiplicative group in $\mathbb{F}_{w+1}$ \cite{Wah1984}. The values $w+1$, such that the AOP of degree $w$ is irreducible is the sequence A001122 in \cite{A001122}. The first values of this sequence are: $ 3, 5, 11, 13, 19, 29, \ldots$. In this paper, we only consider irreducible AOP.

According to Proposition \ref{prop:directProductDecomposition}, $R_{2,w+1}$ is equal to the direct sum of its principal ideals $A_1=((x^{w+1}+1)/p(x))=(x+1)$ and $A_2=((x^{w+1}+1)/x+1)=(p(x))$ and $R_{2,w+1}$ is isomorphic to the direct product of $B_1=\mathbb{F}_{2}[x]/(p(x))=\mathbb{F}_{2^w}$ and $B_1=\mathbb{F}_{2}[x]/(x+1)=\mathbb{F}_{2}$.

It can be shown that the primitive idempotent of $A_1$ is $\theta_1=p(x)+1$. This idempotents is used to build the isomorphism $\phi_1$ between $A_1$ and $B_1$.

\section{Transforms between the field and the ring}
\label{sec:transforms}
This section presents different transforms between the field  $B_1=\mathbb{F}_{2^w}=\mathbb{F}_{2}[x]/(p(x))$ and the ring $R_{2,w+1}=\mathbb{F}_{2}[x]/(x^{w+1}+1)$.

\subsection{Isomorphism transform} 
The first transform is simply the application of the basic isomorphism between $B_1$ and the ideal $A_1$ of $R_{2,w+1}$ (see Prop. \ref{prop:directProductDecomposition}). 

By definition of the isomorphism, we have:
	$$ \phi_1^{-1}( \phi_1(u(x)).\phi_1(v(x)) ) = u(x).v(x) $$   
So, $\phi_1$ can be used to send the elements of the field in the ring, then, to perform the multiplication, and then, to come back in the field. We show in the following Proposition that the isomorphism admits a simplified version.

Let $W(b(x))$, the weight of $b(x)$, defined as the number of monomials in the polynomial representation of $b(x)$.  
\begin{myprop}
	\label{prop:isomorphismForm}
	\begin{equation*}
	\phi_1 (b_B(x)) = b_A(x) = \left \{ \begin{array}{cl}
	b_B(x) & \textrm{if }  W(b_B(x)) \textrm{ is even}  \\
	b_B(x)+p(x) & \textrm{else}
	\end{array}\right.
	\end{equation*}
	
	\begin{equation*}
	\phi_i^{-1}(b_A(x)) = b_B(x) =\left \{ \begin{array}{cl}
	b_A(x) & \textrm{ if }  b_{w}=0 \\
	b_A(x)+p(x) & \textrm{else}
	\end{array}\right.
	\end{equation*}   
	where $b_w$ is the coefficient of the monomial of degree $w$ of $b_A(x)$.
	
\end{myprop} 
\begin{IEEEproof}
	For the first point, we have $\phi_1 (b(x))=b(x)\theta_1(x)=b(x)(p(x)+1)=b(x)p(x)+b(x)$. We can observe that $b(x)p(x)=0$ when $ W(b(x))$ is even and $b(x)p(x)=p(x)$ when $ W(b(x))$ is odd. The first point is thus obvious. 
	
	For the second point, it can be observed that, from the first point of this proposition, if an element of $A_1$ has a coefficient $b_{w}\not=0$, then it was necessarily obtained from the second rule, \textit{i.e.} by adding $p(x)$. Then, its image into $B_1$ can be obtained by subtracting (adding in binary) $p(x)$. If $b_w=0$, then nothing has to be done to obtain $b_B(x)$.
\end{IEEEproof}
\subsection{Embedding transform}
\label{sec:embedding} 
Let us denote by $\phi_E$ the embedding function which simply consists in considering the element of the field as an element of the ring without any transformation. This function was initially proposed in \cite{ITOH1989}.

Note that the images of the elements of $B_1$ doesn't necessarily belong to $A_1$. However, let us define the function $\bar{\phi}_1^{-1}$ from $R_{2,w+1}$ to $A_1$ by $\bar{\phi}_1^{-1}(b_A(x))=b_A(\alpha)$, where $\alpha$ is a root of $p(x)$. This function can be seen as an extension of the function $\phi_1^{-1}$ to the whole ring.  	

\begin{myprop}{\cite{ITOH1989}}
	\label{prop:embedding}
	For any $u(x)$ and $v(x)$ in $B_1$, we have: 
	$$ \bar{\phi}_1^{-1}( \phi_E(u(x)).\phi_E(v(x)) ) = u(x).v(x) $$   
\end{myprop}
\begin{proof}
	The embedding function corresponds to a multiplication by $1$ in the ring. In fact, $1$ is equal to the sum of the idempotents $\theta_i(x)$ of the ideals $A_i$, for $i=1,\ldots,l$ \cite[chapter 8, thm. 7]{MWSl77}. Thus,  $\phi_E(u(x))=u(x).\sum_{i=0}^{l}\theta_i(x)$. Then, $\phi_E(u(x)).\phi_E(v(x))$ is equal to $u(x).v(x).(\sum_{i=0}^{l}\theta_i(x))^2$. Thanks to the properties of idempotents, $\theta_i(x).\theta_j(x)$ is equal to $\theta_i(x)$ if $i=j$ and $0$ else. Thus, $\phi_E(u(x)).\phi_E(v(x))$ is equal to $u(x).v(x).(\sum_{i=0}^{l}\theta_i(x))$. The function $\bar{\phi}_i^{-1}$ is the computation of the remainder modulo $p_i(x)$. The irreducible polynomial $p_i(x)$ corresponds to the ideal $A_i$. Thus $\theta_i(x) \textrm{ mod }p(x)$  is equal to $1$ if $i=1$ and $0$ else. 
\end{proof}

This proposition proves that the Embedding function can be used to perform a multiplication in the ring instead of doing it in the field. The isomorphism also has this property, but the complexities of the transforms between the field and the ring are more complex.

\subsection{Sparse  transform}
\label{sec:sparse} 
Let us define the transform $\phi_S$ from $B_1$ to $R_{2,w+1}$:
\begin{equation*}
	\phi_S (b_B(x)) = b_A(x) = \phi_1(b_B(x)) + \delta . p(x)
\end{equation*}
where $\delta=1$ if $W(\phi_1(b_B(x)) + p(x)) < W(\phi_1(b_B(x))) $ and $0$ else. 
\begin{myprop}
	\label{prop:sparse}
	For any $u(x)$ and $v(x)$ in $B_1$, we have: 
	$$ \bar{\phi}_1^{-1}( \phi_S(u(x)).\phi_S(v(x)) ) = u(x).v(x) $$   
\end{myprop}
\begin{proof}
	 As observed in the proof of Prop. \ref{prop:embedding}, $\bar{\phi}_i^{-1}$ is just the computation of the remainder modulo $p(x)$. Moreover, according to the definition of $\phi_S$, $\phi_S(u(x)).\phi_S(v(x))$ is equal to $u(x).v(x)$ plus a multiple of $p(x)$ (possibly equal to $0$). Thus, the remainder of $\phi_S(u(x)).\phi_S(v(x))$ modulo $p(x)$ is equal to $u(x).v(x)$.  
\end{proof}
This proposition shows that $\phi_S$ can be used to perform the multiplication in the ring. The main interest of this transform is that the weight of the image of $\phi_S$ is small, which reduce the complexity of the multiplication in the ring.
\subsection{Parity  transform}
\label{sec:parity} 

\begin{myprop}
	\label{cor:Parity}
	The ideal $A_1$ is composed of the set of elements of $R_{2,w+1}$ with even weight.
\end{myprop}

\begin{IEEEproof}
We can observe from Proposition \ref{prop:isomorphismForm} that all the image of $\phi_1$ have even weight. Since the number of even-weight element of $R_{2,w+1}$ is equal to the number of elements of $A_1$, $A_1$ is composed of the set of elements of $R_{2,w+1}$ with even weight.
\end{IEEEproof}

Let us consider the function $\phi_P$, from $B_1$ to $R_{2,w+1}$, which adds a single parity bit to the vector corresponding to the finite field element. The obtained element has an even weight (by construction), and thus, according to the previous Proposition, it belongs to $A_1$. 

Since the images by $\phi_P$ of two distinct elements are distinct, $\phi_P$ is a bijection between $B_1$ and $A_1$. The inverse function, $\phi_P^{-1}$, consists just in removing the last coefficient of the ring element.

It should be noted that $\phi_P$ is not an isomorphism, but just a bijection between $B_1$ and $A_1$. However it will be shown in next Section that this function can be used in the context of erasure codes.

\section{Application of transforms}
\label{sec:choice}

In typical \texttt{xor}-based erasure coding systems \cite{xor-luby}, the encoding process consists in multiplying an information vector by the generator matrix. Since in software, \texttt{xor} are performed using machine words of $l$ bits, $l$ interleaved codewords are encoded in parallel. 

We consider a system with $k$ input data blocks and $m$ output parity blocks.

The total number of \texttt{xor} of the encoding is thus defined by the generator matrix which must be as sparse as possible. First, we use a  $k\times (k+m)$-systematic generator matrix built from a $k\times k$-identity matrix concatenated to a $k\times m$ Generalized Cauchy (GC) matrix \cite{generatorMDS:Roth:85}. A GC matrix generates a systematic MDS code and it contains only $1$ on its first row and on its first column. Then, to improve the sparsity of the generator matrix in the ring, we use the Sparse transform $\phi_S$. This has to be done only once since the ring matrix is the same for all the codewords. 

For the information vectors, it is not efficient to use $\phi_S$ since the \texttt{xor}s of machine words do not take into account the sparsity of the \texttt{xor}-ed vectors. We thus use Embedding or Parity transforms, which are less complex than $\phi_1$.

When Embedding is used for information vectors and Sparse is used for the generator matrix, the obtained result in the ring can be sent into the field by using $\phi_1^{-1}$ (proof similar to the proofs of Propositions \ref{prop:embedding} or \ref{prop:sparse}). 

When Parity is used for the information vector, the image of the vector in the ring only contains elements of the ideal $A_1$. Since these elements are multiplied by the generator matrix (in the ring), the obtained result only contains elements of the ideal $A_1$. These elements have even weight, so it is not necessary to keep the parity bit before sending them on the "erasure channel". Since Parity transform is not an isomorphism, these data can not be decoded by another method. Indeed, to decode, it is necessary to apply $\phi_P$ (add the parity bit), then to decode by multiplying by the inverse matrix, and then to to apply $\phi_P^{-1}$ (remove the parity bit on the correct information vector).

\section{Complexity analysis}
\label{sec:complexity}
In this section, we determine the total number of \texttt{xor} operations done in the coding and the decoding processes. 

\subsection{Coding complexity}
\label{sec:codingComplexity}
The coding process is composed of three phases: the field to ring transform, the matrix vector multiplication and the ring to field transform. We assume that the information vector is a vector of $k$ elements of the field $\mathbb{F}_{2^w}$.

For the first and the third phases, Table \ref{tab:compAOP} gives the complexities of Embedding and Parity transforms obtained from their definition in Section \ref{sec:transforms}. 
\begin{table}[h]
	\centering 
	\begin{tabular}{c||c|c}
		 & field to ring & ring to field \\
		\hline
		\hline
		Embedding & 0 & $m.w$ \\
		\hline
		Parity bits & $k.w$ & 0 \\
	\end{tabular}	
	\caption{\label{tab:compAOP} Number of xor for Embedding and Parity transforms}
\end{table}	

The choice between the two methods thus depends on the values of the parameters: if $k>m$, Parity transform has lower complexity. Else, "Embedding" complexity is better.  

For the matrix vector operation, let us first consider the multiplication of two ring elements. As explained in the previous section, the first element (which corresponds to an information symbol) is managed by the software implementation by machine words. So the complexity of the multiplication only depends on the weight of the second element, denoted by $w_2\in\{0, 1, \ldots,w+1\}$. The complexity of this multiplication is thus $(w+1).w_2$.  

Now, we can consider the specificities of the various transforms. In the Parity transform, the last bit of the parity blocks is not used ( i.e. it is not transmitted on the erasure channel). So it is not necessary to compute it. It follows that the complexity of the multiplication is only $w.w_2$. 

Similarly, for the Embedding transform, the last bit of the input vector is always equal to $0$. So, we also have a complexity equal to $w.w_2$. 


To have an average number of operations done in the multiplication of the generator matrix by the input data blocks, we have to evaluate the average weight of the entries of the generator matrix in the ring. 

The generator matrix is a $k\times m$-GC matrix with the first column and the first row are filled by $1$.  The other elements can be considered as random non zero elements. They are generated by $\phi_S$ which chooses the lowest ring element among the two ones corresponding to the field element. Let us denote their average weight by $w_{\phi_S}$. 
For this case, the average number of \texttt{xor}s is thus:
$$ (k+m-1).w + (k-1).(m-1).w.w_{\phi_S} $$  

This leads to the following general expression of the coding complexity:
$$ (\min(k,m)+k+m-1).w+(k-1).(m-1).w.w_{\phi_S} $$ 




To estimate the complexity on a practical example, we fix the value of $w$ to $4$. Classic combinatorial evaluation (not presented here) gives the average weight for nonzeros images of $\phi_S$: $$w_{\phi_S}=\frac{w+1}{2^{w+1}-2}.\Big(2^w - \binom{w}{w/2}\Big)$$ So, $w_{\phi_S}=1.66$. We plot in Figure \ref{fig:factor} the evolution of the factor over optimal (used \textit{e.g.} in \cite{plankScheduling:2011}, table III) which is the density of the matrix normalized by the minimal density, $k.m.w$. We vary the value of $k$ for three values of $m$: $3$, $5$ and $7$. For each pair $(k,m)$, we generate 10000 random GC matrices and keep the best we found.

\begin{figure}[h]
	
	\begin{center}
		
		\begin{tikzpicture}		
		\begin{axis}[		
		ymajorgrids,		
		width=0.48\textwidth,		
		height=0.20\textheight,		
		xtick pos=left,		
		ytick pos=left,		
		legend columns=3,		
		legend entries={m=3,m=5,m=7, best m=3, best m=5, best m=7},		
		legend to name=named,		
		axis x line*=bottom,		
		axis y line*=left,		
		ymin=1,
		ymax=1.8,
		ytick={1,1.2,1.4, 1.6, 1.8},
		xlabel=$k$,
		ylabel=factor,
		title={\textbf{factor over optimal}}]
		\addplot+[red, mark=*,  mark size=1.5pt,mark options={fill=red}] table [ x=x, y=a,, col sep=space] {factor_m3.csv};			
		\addplot+[blue, mark=*,  mark size=1.5pt,mark options={fill=blue}] table [ x=x, y=a,, col sep=space] {factor_m5.csv};			
		\addplot+[green!40!black, mark=*,  mark size=1.5pt,mark options={fill=green!40!black}] table [ x=x, y=a,, col sep=space] {factor_m7.csv};
		\addplot+[solid, red!40!white, mark size=1.5pt, mark=triangle*, mark options={fill=red!40!white}] table [ x=x, y=a,, col sep=space] {factor_m3_min.csv};			
		\addplot+[solid, cyan, mark size=1.5pt, mark=triangle*, mark options={fill=cyan}] table [ x=x, y=a,, col sep=space] {factor_m5_min.csv};			
		\addplot+[solid, green, mark size=1.5pt, mark=triangle*, mark options={fill=green}] table [ x=x, y=a,, col sep=space] {factor_m7_min.csv};				
		\end{axis}		
		\end{tikzpicture}
		\ref{named}
		\caption{Factor over optimal depending on $m$}	
		\label{fig:factor}		
	\end{center}
\end{figure}
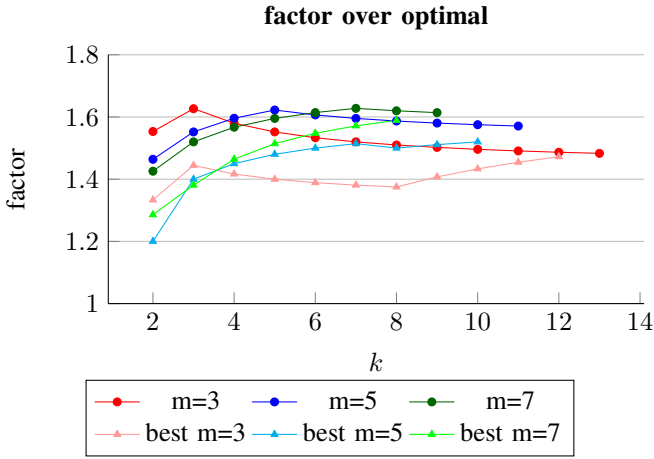

We can observe that the values are very low. For example, \cite{plankScheduling:2011} gives the lowest density of Cauchy matrices for the field $\mathbb{F}_{2^4}$ and we can observe that our values are always lower than these ones. 



To reduce the complexity in specific cases, we can observe that the ring contains $w$ elements whose the corresponding matrix is optimal (one diagonal). By using these elements, we can search by brute force MDS matrices built only with these optimal elements.
For example, let us consider the elements of the field $\mathbb{F}_{2^4}$ sent into the ring $R_{2,5}$. The Vandermonde matrix defined by:
$$ V = \big( x^{i.j}\big)_{i=0,\ldots,4; j=0,\ldots,4} $$  
where $x$ is a monomial in $R_{2,5}$ has the minimal number of $1$. It can be verified that this matrix can be used to build a systematic a MDS code. For this matrix, the total number of \texttt{xor}s done in the generation of the parity packets (including the field-to-ring and ring-to-field transforms) is 
$$ (\min(k,m)+k+m-1).w+(k-1).(m-1).w = k.w + k.m.w$$
Its factor over optimal is equal to $1.2$ which is lower than the values given in Figure \ref{fig:factor} and which is close to the lowest bound given in \cite{BlaumR99}. 



\subsection{Decoding complexity}
\label{sec:decodingComplexity}

As the decoding is a matrix inversion and a matrix vector multiplication, we can use the same approach to perform the multiplication. We first invert the sub-matrix in the field, then we transform each entry of this matrix into ring elements. Then, we perform the ring multiplication.

The complexity of the decoding thus depends on the complexity of the matrix inversion and on the complexity of the matrix vector multiplication. 

The complexity of the matrix vector multiplication was studied in the previous paragraph.

The complexity of the a $r\times r$-matrix inversion is generally in $O(r^3)$ operations in the field. But if the matrix has a Cauchy structure, this complexity can be reduced to $O(r^2)$ \cite{xor-luby}.

Note that, contrarily to the matrix vector multiplication, the matrix inversion complexity does not depend on the size of the source and parity blocks. And thus, it becomes negligible when the size of the blocks increase.

\section{Scheduling}
\label{sec:scheduling}

An interesting optimization on MDS erasure codes under \texttt{xor}-based representation is the scheduling of \texttt{xor} operations.

Such techniques are proposed in \cite{Hafner:2005}, \cite{Huang:2007}, \cite{plankScheduling:2011} and \cite{plankDSN2012}. The general principle consists in "factorizing" some \texttt{xor} operations which are done several times to generate the parity blocks. 

We show in the two next paragraphs that these techniques can be used very efficiently on the ring elements. 

However, it can be observed that the matrices defined over rings have two main advantages.

\subsection{Complexity reduction}
\label{sec:complexityScheduling}

Over finite fields, the scheduling consists in searching common patterns on the binary representation of the generator matrices. The $w \times w$-matrices representing the multiplication by the field elements does not have particular structure and thus, they must be entirely considered in the scheduling algorithm. 

This is not the case for the $(w+1)\times (w+1)$-matrices corresponding to a ring element because, thanks to the form of the polynomial $x^{w+1}+1$, they are composed of diagonals either full of $0$ or $1$. This means that they can be represented in the scheduling algorithm just by their first column or, equivalently, by the ring polynomial.

This allows to drastically reduce the algorithm complexity and thus to handle bigger matrices. From a polynomial point of view, the search of scheduling just consist in finding some common patterns in the equations generating the parity blocks.

\begin{example}
	Let us assume that $n=5$ and that three data polynomials $a_0(x),\ a_1(x)$ and $a_2(x)$ are combined to generate the three parities $p_0(x)=(1+x^4)a_0(x)+x^2a_1(x)+x^3a_2(x)$,  $p_1(x)=a_0(x)+x^3a_1(x)+(1+x^3)a_2(x)$ and $p_2(x)=a_0(x)+a_1(x)+x^3a_2(x)$. 

	In this case, the scheduling just consists in computing  $p'(x)= a_0(x)+x^3a_2(x) $ and then $p_0(x)=p'(x)+x^4a_0(x)+x^2a_1(x)$, $p_1(x)= p'(x)+x^3a_1(x)+a_2(x)$ and $p_2=p'(x)+a_2(x)$.	
	
	To estimate the complexity, we can consider the number of  sums of polynomials. Without scheduling, we need $11$ sums ($4$ for $p_0(x)$, $4$ for $p_1(x)$, and $3$ for $p_2(x)$) instead of with scheduling, we only need 10 sums ($2$ for $p'(x)$, $3$ for $p_0(x)$, $3$ for $p_1(x)$ and $2$ for $p_2(x)$).   
\end{example}

\subsection{Additional patterns}
\label{sec:additionalPatterns}
Ring-based matrices allow to find more common patterns than field-based matrices. The main idea is to observe that, in the ring, we can "factorize" not only common operations, but also operations which are multiple by a monomial (\emph{i. e.} cyclic-shift) of operations done in some other equations. This is possible only because the multiplications are done modulo $x^{w+1}+1$.

\begin{example}
	Let us assume that $n=5$ and that three data polynomials $a_0(x),\ a_1(x)$ and $a_2(x)$ are combined to generate the parities 
	$p_0(x)=a_0(x)   + x^2a_1(x)+(1+x^2)a_2(x)$,  
	$p_1(x)=x^2a_0(x)+ x^3a_1(x)+(x+x^4)a_2(x)$ and 
	$p_2(x)=x^2a_0(x)+ a_1(x)+(x^2+x^3)a_2(x)$. 
	We can observe that, with a "simple" scheduling, it is not possible to factorize some operations. 
	
	However, by rewriting the polynomials, we can reveal factorizations: $p_1(x)=x^2a_0(x)+x(x^2a_1(x)+a_2(x))+x^4a_2(x)$ and $p_2(x)=x^2a_0(x)+x^3(x^2a_1(x)+a_2(x))+x^2a_2(x)$. So, if $p'(x)=x^2a_1(x)+a_2(x)$, we have $p_0(x)=p'(x)+a_0(x)+x^2a_2(x)$, $p_1(x)=xp'(x)+x^2a_0(x)+x^2a_2(x)$ and $p_2(x)=x^3p'(x)+x^2a_0(x)+x^2a_2(x)$. 
	
	To estimate the complexity by the same method than in the previous example, we need $11$ polynomial additions with scheduling compared to $12$ additions necessary without scheduling. 		 		
\end{example}

\subsection{Scheduling results}
\label{sec:schedulingResults}
To evaluate the potential gain of the scheduling, we have implemented an exhaustive search of the best patterns on generator matrices. 

This algorithm was applied on several codes for the field $\mathbb{F}_{2^4}$. Table \ref{tab:factorScheduling} presents the results in term of "factor over optimal" which is defined as the total number of 1 in the matrix over the number of 1 for the optimal MDS matrix , i.e $k.m.w$.

When working in a ring, we include to the complexity the operations needed to apply the transforms. In this case, Embedding transform has a lower complexity. So we added $m.w$ to the number of 1 in the matrix resulting from the scheduling algorithm.

For each case, we have generated $100$ random Generalized Cauchy matrices. 

The measured parameters are:
\begin{itemize}
	\item average field matrix: average number of $1$ in the GC matrices divided by $k.m.w$
	\item best field matrix: lowest number of $1$ among the GC matrices divided by $k.m.w$
	\item average ring matrix:  average number of $1$ in the ring matrices (without scheduling) + ring-field correspondence  divided by $k.m.w$
	\item best ring matrix:  best number of $1$ among the ring matrices (without scheduling) + ring-field correspondence  divided by $k.m.w$
	\item average with scheduling: average number of \texttt{xor}s with scheduling + ring-field correspondence  divided by $k.m.w$
	\item best with scheduling: best number of \texttt{xor}s with scheduling + ring-field correspondence  divided by $k.m.w$	
\end{itemize} 

\begin{table}[h]
	\centering 
	\begin{tabular}{c||c||c}
		k+m,k	 & \multicolumn{1}{c||}{12,8} & \multicolumn{1}{c}{ 16,10} \\
		\hline \hline
		average field matrix & 1.79  & 1.90   \\
		\hline
		best field matrix & 1.73   & 1.85  \\
		\hline
		average ring matrix & 1.59  & 1.63 \\
		\hline
		best ring matrix & 1.5  & 1.58 \\
		\hline
		average with scheduling  & 1.32 &  1.26  \\
		\hline
		best with scheduling & 1.19 &  1.20  \\
	\end{tabular}
	\caption{\label{tab:factorScheduling} factor over optimal for $w=4$}
\end{table}	

This table confirms that, even without scheduling, ring matrices have a lower density than field matrices, thanks to the Sparse transform.
Applying scheduling to these matrices allows a significant gain of complexity. Indeed, it reduces the complexity by more than $20\%$  on the best matrices. The final results are similar to the results obtained (without scheduling) on the optimal matrix in Section \ref{sec:codingComplexity}. To the best of our knowledge, other scheduling approaches do not reach this level of sparsity for these parameters.

\section{Conclusion}
In this paper, we have presented a new method to build MDS erasure codes with low complexity. By using transforms between a finite field and a polynomial ring, sparse generator matrices can be obtained. This allows to significantly reduce the complexity of the matrix vector multiplication. 
It also allows simple schedulers that drastically improve the complexity by reducing the number of operations. 
\\Similar results can be obtained with Equally-Spaced Polynomials (ESP) \cite{ITOH1989}, but they are not presented here due to lack of space.

\bibliographystyle{IEEEtran}
\bibliographystyle{plain}
	\bibliography{biblio}

\end{document}